\documentclass[preprint,11pt]{elsarticle}
\usepackage{lipsum}
\makeatletter
\def\ps@pprintTitle{%
 \let\@oddhead\@empty
 \let\@evenhead\@empty
 \def\@oddfoot{}%
 \let\@evenfoot\@oddfoot}
\usepackage{geometry}
\usepackage{graphicx}
\usepackage{amssymb}
\usepackage{lineno}
\usepackage{algorithm}
\usepackage{algorithmic}
\usepackage{float}
\usepackage{caption}%
\usepackage{amsmath}
\usepackage{siunitx}
\usepackage{amsfonts}
\usepackage{subfig}
\usepackage{booktabs}
\usepackage{placeins}
\usepackage{xr}
\usepackage{amsthm}
\usepackage{bbding}
\usepackage{array}
\usepackage{natbib}
\usepackage{relsize}
\usepackage{subfig}
\usepackage{graphicx}
\usepackage{esint}
\usepackage{tikz}
\usepackage{arydshln}
\usepackage{eurosym}
\usepackage{color}
\usepackage{multirow}
\usepackage{textcomp}
\usepackage{tabularx}
\usepackage{placeins}
\usepackage{lscape}
\usepackage{amsmath}
\usepackage{leftidx}
\usepackage{float}
\usepackage{diagbox}
\usepackage{amssymb}
\usepackage{pgf,tikz,pgfplots}
\pgfplotsset{compat=1.14}
\usepackage{mathrsfs}
\usetikzlibrary{arrows}
\floatstyle{plaintop}
\restylefloat{table}
\usepackage{array}
\newcolumntype{L}[1]{>{\raggedright\let\newline\\\arraybackslash\hspace{0pt}}m{#1}}
\newcolumntype{C}[1]{>{\centering\let\newline\\\arraybackslash\hspace{0pt}}m{#1}}
\newcolumntype{R}[1]{>{\raggedleft\let\newline\\\arraybackslash\hspace{0pt}}m{#1}}

\newtheorem{lemma}{Lemma}
\newtheorem{theorem}{Theorem}

\newtheorem{definition}{Definition}
\usepackage[hidelinks=True,
            colorlinks = true,
            linkcolor = blue,
            urlcolor  = blue,
            citecolor = blue,
            anchorcolor = blue]{hyperref}
\setcitestyle{open={(},authoryear,close={)}}

\journal{SODA}

\begin{document}
\definecolor{dtsfsf}{rgb}{0.8274509803921568,0.1843137254901961,0.1843137254901961}
\definecolor{rvwvcq}{rgb}{0.08235294117647059,0.396078431372549,0.7529411764705882}
\definecolor{dtsfsf}{rgb}{0.8274509803921568,0.1843137254901961,0.1843137254901961}
\definecolor{sexdts}{rgb}{0.1803921568627451,0.49019607843137253,0.19607843137254902}
\definecolor{rvwvcq}{rgb}{0.08235294117647059,0.396078431372549,0.7529411764705882}
\definecolor{wrwrwr}{rgb}{0.3803921568627451,0.3803921568627451,0.3803921568627451}
\begin{frontmatter}

\author[a,b]{Mojtaba Abdolmaleki\corref{cor1}}
\author[a,b]{Yafeng Yin}
\author[a]{Neda Masoud}

\address[a]{Department of Civil and Environmental Engineering, University of Michigan, Ann Arbor, MI 48109, United States}

\address[b]{Department of Industrial and Operations Engineering, University of Michigan, Ann Arbor, MI 48109, United States}


\title{Minimum Weight Pairwise Distance Preservers}



\begin{abstract}

In this paper, we study the Minimum Weight Pairwise Distance Preservers (MWPDP) problem. Consider a positively weighted undirected/directed connected graph $G = (V, E, c)$ and a subset $P$ of pairs of vertices, also called demand pairs. A subgraph $G'$ is a distance preserver with respect to $P$ if and only if every pair $(u, w) \in P$ satisfies $dist_{G'} (u, w) = dist_{G}(u, w)$. In MWPDP problem, we aim to find the minimum-weight subgraph $G^*$ that is a distance preserver with respect to $P$. Taking a shortest path between each pair in $P$ gives us a trivial solution with the weight of at most $U=\sum_{(u,v) \in P} dist_{G} (u, w)$. Subsequently, we ask how much improvement we can make upon $U$. In other words, we opt to find a distance preserver $G^*$ that maximizes $U-c(G^*)$. Denote this problem as Cost Sharing Pairwise Distance Preservers (CSPDP), which has several applications in the planning and operations of transportation systems.

The only known work that can provide a nontrivial solution for CSPDP is that of Chlamt{\'a}{\v{c}} et al. (SODA, 2017). This algorithm works for unweighted graphs and guarantees a non-zero objective only if the optimal solution
is extremely sparse with respect to the trivial solution. We address this issue by proposing an $O(|E|^{1/2+\epsilon})$-approximation algorithm for CSPDP in weighted graphs that runs in $O((|P||E|)^{2.38} (1/\epsilon))$ time. Moreover, we prove CSPDP is at least as hard as $\text{LABEL-COVER}_{\max}$. This implies that CSPDP cannot be approximated within $O(|E|^{1/6-\epsilon})$ factor in polynomial time, unless there is an improvement in the notoriously difficult $\text{LABEL-COVER}_{\max}$.

\end{abstract}

\cortext[cor1]{Corresponding author. E-mail address: mojtabaa@umich.edu}
    \end{frontmatter}

\newpage
\section{Introduction}

\noindent A fundamental question in graph theory is to sparsify an input graph by finding a subgraph that has fewer edges or less total weight, while preserving some specific properties of the original graph. In the pairwise distance preservers problem, one is given a connected undirected/directed graph $G = (V, E)$ and a subset $P$ of pairs of vertices, also called demand pairs. The aim is to find a subgraph, $G^*$, with a minimum number of edges, that preserves the exact distance in the original graph, $dist_{G^*} (u, w) =  dist_{G}(u, w)$ for each demand pair in $P$ \citep{coppersmith2006sparse}. A large body of research is devoted to pairwise distance preservers. For recent work, see, e.g.,  \citet{abboud2016error,bodwin2016better,bodwin2017linear,chlamtavc2017approximating,bodwin2019structure,chang2018near,gajjar2017distance}. Pairwise distance preservers have application in research on several other related theoretical problems such as spanners, distance oracles, graph algorithms, etc \citep{abboud2016error,abboud20174,abboud2018hierarchy,alon2002testing,bodwin2015very,bodwin2016better,bollobas2005sparse}.

Similar to the work of \cite{elkin2007hardness}, we generalize the concept of pairwise distance preservers into a weighted version, Minimum Weight Pairwise Distance Preserver (MWPDP), where the input is an undirected/directed weighted graph $G = (V, E, c, l)$ and set of demand pairs $P$. Here, $c: E \rightarrow R^+$ and  $l: E \rightarrow R^+$  are the wight function and the length function, respectively. The goal is to find a subgraph, $G^*=(V^*,E^*)$, that minimizes $c(G^*)=\sum_{e \in E^*} c(e)$, while preserves the exact distance in the original graph, $dist_{G^*} (u, w) =  dist_{G}(u, w)$. We focus on the case where the weight and length functions coincide, i.e. $c(e)=l(e) : \forall e \in E$. This was also the focus of a number of other studies in spanners problems \citep{althofer1990generating,awerbuch1992efficient,chandra1992new, regev1995weight}.

This paper intends to study the Minimum Weight Pairwise Distance Preserver (MWPDP) problem from an optimization point of view. We can first obtain a trivial upper bound, $U=\sum_{(u,v) \in P} dist_{G} (u, w)$, by finding a shortest path for each pair of nodes in $P$ and considering their union as the subgraph $G'$. It becomes natural to ask about the amount of improvement that we can make upon $U$. In other words, is there any efficient algorithm that can find a distance preserver $G^*=(V,E^*)$ that maximizes $U-|c(E^*)|$? Denote this problem as Cost Sharing Pairwise Distance Preservers (CSPDP) problem. CSPDP has direct applications in transportation systems such as the planning of vehicle platooning  \citep{abdolmaleki2019itinerary, luo2020repeated,sethuraman2019effects}. The improvement $U-|c(E^*)|$ quantifies the potential profit from platooning, which will take some cost to form and maintain. As such, it is vital to guarantee that this improvement is higher than the required setup cost.

The only known work that can provide a nontrivial solution for CSPDP is the work of \cite{chlamtavc2017approximating} on MWPDP instances with $c(e)=l(e)=1, \forall e \in E$, which provides an $O(n^{3/5+\epsilon})$-approximation. It cannot guarantee to obtain a non-zero objective unless the optimal solution, $G^*$, is extremely sparse with respect to the trivial upperbound $U$. More specifically, the optimal solution should satisfy $\frac{|E^*|}{U} \leq \frac{1}{n^{3/5}}$. We address this issue by proposing an $O(m^{1/2+\epsilon})$-approximation algorithm for CSPDP in weighted graphs that runs in $O((|P|m)^{2.38} (1/\epsilon))$ time, where $m=|E|$ is the number of edges in graph $G$. 

\begin{theorem}\label{thm:final}
For any constant $\epsilon > 0$, there is a
polynomial-time $O(m^{1/2+\epsilon})$-approximation algorithm for any instance of CSPDP problem.
\end{theorem}

\begin{proof}
See Section \ref{sec:appxalg}.
\end{proof}

On the other hand, we prove CSPDP is at least as hard as $\text{LABEL-COVER}_{\max}$. This implies that CSPDP cannot be approximated within $O(|E|^{1/6-\epsilon})$ factor in polynomial time, unless there is an improvement on the notoriously difficult $\text{LABEL-COVER}_{\max}$.

\begin{theorem}\label{thm:hardnessfinal}
For any constant $\epsilon > 0$, there no polynomial-time approximation algorithm for CSPDP achieving a ratio of $2^{\log^{1-\epsilon}m^{1/2}}$, for any $0< \epsilon< 1$, unless $NP \subseteq DTIME(n^{polylog(n)})$.
\end{theorem}
\begin{proof}
See Section \ref{sec:hardness}.
\end{proof}

The rest of the paper is organized as follows: Section \ref{sec:probover} provides a reduction from the undirected version of CSPDP into its directed version. Section \ref{sec:algov} provides an overview of the CSPDP problem as well as a description for the algorithms devised in this paper. Afterwards, Section \ref{sec:appxalg} analyzes the algorithms both from the approximation performance and computational complexity point of view. Finally, Section \ref{sec:hardness} presents the hardness result of CSPDP for both directed and undirected versions of the problem.

\section{Problem overview}\label{sec:probover}

\noindent In this section, we begin with reducing the undirected CSPDP into its dircted CSPDP. Then, we devise a natural integer programming formulation for CSPDP that will be used in our approximation algorithm. Observe that we cannot reduce the undirected CSPDP into its directed counterpart by simply substituting each edge $e=(v_i,v_j)$ with two directed edges $\vec{e}_1=(v_i,v_j)$ and $\vec{e}_2=(v_j,v_i)$. The following lemma describes a reduction from the undirected version into its directed version.

\begin{lemma}\label{lemma:redundi}
For any instance of CSPDP with pair set $P$ in an undirected weighted graph $G$ with $n$ vertices and $m$ edges, we can construct an instance of CSPDP with the same pair set $P$ in a directed weighted graph $R$ with $n+2m$ vertices and $5m$ edges and the same objective value.

\end{lemma}
\begin{proof}
Given an undirected graph $G$, delete each edge $e=(v_i,v_j) \in E$. Then, add two vertices $v'_{ij}$ and $v^{''}_{ij}$, and five directed edges namely, $(v_i,v'_{ij}), (v_j,v'_{ij}), (v^{''}_{ij},v_i), (v^{''}_{ij},v_j) $ with weight 0, and $(v'_{ij},v^{''}_{ij})$ with weight $c_{v_i,v_j}$. Denote by graph $R=(V_1,E_1)$ the new graph obtained after applying the same procedure on all edges $e\in E$. Figure \ref{fig:reduction} describes the procedure applied on an arbitrary edge $(v_i,v_j) \in E$. 

We prove that finding the optimal solution $G^*$ for CSPDP in undirected graph $G$ with demand pairs $P$ is equivalent to finding the optimal solution $R^*$ for CSPDP in directed graph $R$ with a directed version of $P$. As $R$ is a directed graph, for each demand pair $(s,t) \in P$ we consider one of them as the origin and the other one as the destination.

First, note that the distance between pairs of nodes $(s,t) \in V \times V$ in $G$ will be preserved in the graph $R$. As such, in both $G$ and $R$ the trivial upper bound $U$ takes the same value.

For any feasible solution $G^*$ for CSPDP in undirected graph $G$, we find a feasible solution for CSPDP in directed graph $R$ with the same objective value. To reach this goal, for all edges $e=(v_i,v_j) \in G^*$, we take the union of the edges $(v_i,v'_{ij}), (v_j,v'_{ij}), (v^{''}_{ij},v_i), (v^{''}_{ij},v_j)$ and $(v'_{ij},v^{''}_{ij})$ to form the subgraph $R^*$. Conversely, for any solution $R^*$ for CSPDP in directed graph $R$, we find a feasible solution for CSPDP in the original undirected graph $G$ with the same objective value. To do so, take the union of all edges $e=(v_i,v_j) \in E$  such that $(v'_{ij},v^{''}_{ij}) \in E_1$.

\begin{figure}
    \centering
    \includegraphics[width=0.6\linewidth]{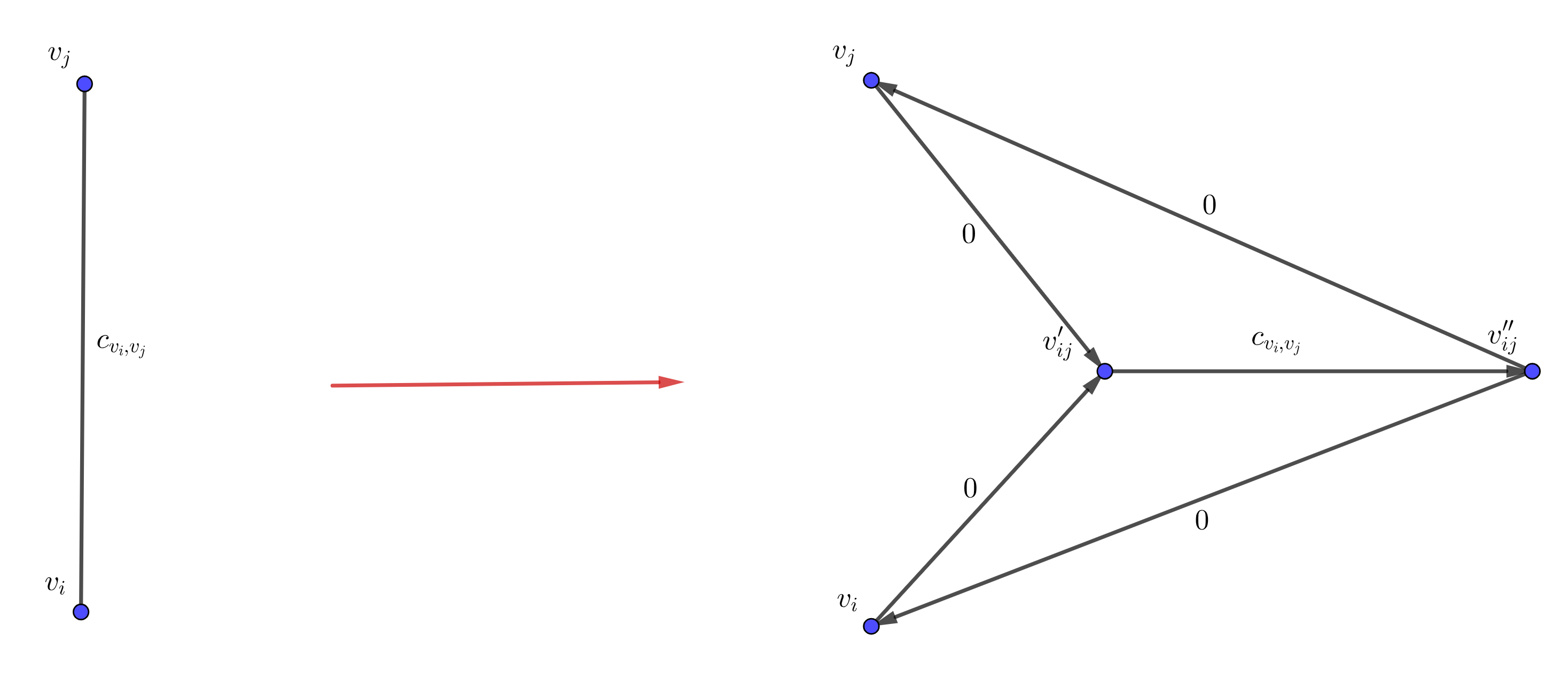}%
    \caption{reduction}
    \label{fig:reduction}
\end{figure}

\end{proof}

Given the results in Lemma \ref{lemma:redundi}, it is sufficient for us to propose our approximation algorithm for the directed version of the problem.

Consider a directed graph $G=(V,E)$, and an assignment of weights to the edges $c : E \rightarrow R^{+}$. For any pair $(s,t) \in P$, let $Q_{(s,t)}$ denote the collection of all directed shortest paths in $G$ from $s$ to $t$, and define the local graph $G_{(s,t)} =(V_{(s,t)}, E_{(s,t)})$ as the union of all nodes and edges in $Q_{(s,t)}$.

Consider an optimal solution $H^*$ to CSPDP, for each pair of nodes $(s,t) \in P$ consider a shortest path $p_{s,t}$ in the subgraph $H^*$ connecting the pair of nodes $(s,t)$. Optimality of $H^*$ implies that each edge with positive weight in $H^*$ appears in at least one of these selected paths.

Now, for each edge $e \in H^*$ denote the number of shortest paths $p_{(s,t)}: (s,t) \in P$ containing edge $e$ by $u_e$; it is trivial that the upper bound $U$ can be decomposed into the following summation: 
\begin{equation*}
    U= \sum_{(u,v) \in P} dist_{G} (u, w)= \sum_{e \in H} u_e c(e)
\end{equation*}
As such, we can write the objective in CSPDP as:
\begin{equation*}
    U-c(G')=\sum_{e \in H} (u_e-1)c(e)
\end{equation*}

We are now ready to formulate the CSPDP problem as the following integer programming:

\begin{subequations}\label{opt:pdps}
\begin{align}\label{obj:pdps}
\mbox{Max} \quad & z_1 = \sum_{e\in E} ( \sum_{(s,t):e\in E_{(s,t)}} x^{(s,t)}_e -y_e)c(e) \\ \label{conspdps:assign}
\mbox{s.t.} \quad & \sum_{\substack{ v_i: \\ e=(v_i,v)\in E(s,t)}} x_e^{(s,t)}- \sum_{\substack{ v_j: \\ e=(v,v_j)\in E(s,t)}} x_e^{(s,t)}=d_v^{(s,t)} && \forall v \in E, \forall (s,t) \in P
\\   \label{cons1:vector}
& x_e^{(s,t)} \leq y_e   &&\scriptstyle{ \forall e \in E(G), \forall (s,t): e \in E_{(s,t)} }
\\  \label{conspdps:integ}
& x_{e}^{(s,t)} \in \left \{0,1  \right \}, &&\substack{\forall (s,t)\in P, \\ \forall e \in E(s,t)}
\\  \label{conspdps:integaux}
& y_e \in \left \{0,1  \right \}, &&{\scriptstyle \forall e \in E}
\end{align}
\end{subequations}
Where $d_v^{(s,t)}$ is defined as follows:
\begin{equation*}
    d_v^{(s,t)}=   \left\{\begin{matrix}
-1 \quad \textnormal{If  } \textit{  s=v}\\ 
\ \ 1\quad \textnormal{If   } \textit{ t=v}\\
0 \quad \textrm{Otherwise}
\end{matrix}\right.
\end{equation*}

Let us define savings for an edge $e$ as $( \sum_{(s,t):e\in E_{(s,t)}} x^{(s,t)}_e -y_e)$. The objective \eqref{obj:pdps} is the summation of the savings in all edges of $G$. Constraint \eqref{conspdps:assign} ensures that there exists a shortest path from $s$ to $t$ for each pair $(s,t) \in P$. The binary decision variable $x^{(s,t)}_e$ takes value 1 if the directed shortest path for pair $(s,t)$ contains edge $e$ and  0 otherwise. Constraint \eqref{cons1:vector}  defines $y_l$ as a binary variable that is equal to 1 if the edge $e$ is contained by at least one of the directed shortest paths $p_{(s,t)}$ in a feasible solution.

\section{Algorithm}\label{sec:algov}

\noindent In this section, we propose an $O(m^{1/2+\epsilon})$-approximation Algorithm for CSPDPD problem for any $\epsilon>0$. We save the analysis for Section \ref{sec:appxalg}. 


\begin{definition} \label{def:thickedge}
Let $v_e$ be the number of local graphs $G_{(s,t)}$ containing edge $e$. We say an edge $e \in E(G)$ is thick if $v_e \geq \frac{|P|}{m^{1/2 +\epsilon}}$. We say an edge is thin otherwise. Moreover, denote by $E_{TK}$ the set of thick edges and by $E_{TN}$ the set of thin edges.
\end{definition}

Let us consider an optimal solution $(x^{opt},y^{opt})$ to problem \eqref{opt:pdps}. Then, for the objective \eqref{obj:pdps}, we have two possible cases depending on whether condition \eqref{condition:main} holds or not:

\begin{equation}\label{condition:main}
    \sum_{e\in E_{TK}} ( \sum_{(s,t):e\in E_{(s,t)}} x^{opt,(s,t)}_e -y^{opt}_e)c(e)>\frac{1}{m^{\epsilon}} \sum_{e\in E} ( \sum_{(s,t):e\in E_{(s,t)}} x^{opt,(s,t)}_e -y^{opt}_e)c(e)
\end{equation}

\begin{definition}
We say an instance $(G,P)$ is thick-dominant if condition \eqref{condition:main} holds for at least one optimal solution. We call that optimal solution a thick-dominant solution. We say an instance is thin-dominant otherwise. 
\end{definition}

We take two different approaches to approximate the optimal solution for thick-dominant and thin-dominant instances of CSPDP. Specifically, we approximate the optimal solution of a thick-dominant instance by forcing all paths to have an overlap with a centric path, while, for a thin-dominant instance, we approximate its optimal solution with feasible solutions that consist of shortest paths whose overlaps are almost evenly distributed over the edges of $G$.

For thick-dominant instances, we propose the dynamic-programming-based solution in algorithm \ref{alg:dyn} to solve the problem. The core idea of algorithm \ref{alg:dyn} is to find a demand pair $(s^*,t^*)$ and a shortest path $p_{(s^*,t^*)}$ from $s^*$ to $t^*$ that will maximize the objective \eqref{obj:pdps} when we force every other demand pair $(s',t')$ to pass through the overlap of its restricted subgraph $G_{(s',t')}$ and path $p_{(s^*,t^*)}$.


\vspace{0.5cm}
\begin{algorithm}[H]
    \caption{Dynamic programming Algorithm for thick-dominant instances}
    \label{alg:dyn}
\vspace{0.5cm}
\begin{algorithmic}

 \STATE 1. For each demand pair $(s,t)\in P$, consider a weighted copy of graph $G_{(s,t)}$ and denote it by $H_{(s,t)}$.
 \STATE 2. Set a weight of $w(e)=(v_e-1)c(e)$ for each thick edge $e \in H_{(s,t)}$ and a weight of zero for each thin edge $e \in H_{(s,t)}$.
 \STATE 3. Find the maximum weighted path, using dynamic programming, for each pair of nodes $(s,t) \in P$ in their corresponding local graph $H_{(s,t)}$. Denote the longest path by $p_{(s,t)}$ and its weight by $a_{(s,t)}=\sum_{e \in p_{(s,t)}}w(e)$. 
 \STATE 4. Denote the maximum weighted path in set $\left \{  p_{(s,t)}: (s,t) \in P \right \}$ by $p_{(s^*,t^*)}$ and its weight by $a^*=\sum_{e \in p_{(s^*,t^*)}}w(e)$.
 \STATE 5. For each demand pair $(s,t)\in P$, consider a new weighted copy of graph $H_{(s,t)}$ and denote it by $H^*_{(s,t)}$. 
 \STATE 6. Set a weight of $w(e)=c(e)$ for each edge $e \in H^*_{(s,t)}, e \in p_{(s^*,t^*)}$ and a weight of zero for each edge $e \in H^*_{(s,t)}, e \notin p_{(s^*,t^*)}$.
 \STATE 7. Find the maximum weighted path, using dynamic programming, for each pair of nodes $(s,t) \in P$ in their corresponding local graph $H^{*}_{(s,t)}$ and denote it by $h_{(s,t)}$.
 \STATE 8. Denote by $H$ the subgraph obtained by taking the union of all the paths $h_{(s,t)}$.
 \STATE 9. Output $H$ and $c(H)$.
\end{algorithmic}
\end{algorithm}
\vspace{0.5cm}

In a thin-dominant instance, in the optimal solution the portion of the objective \eqref{obj:pdps} from the thick edges are negligible. On the other hand, each thin edge is on the overlap of a few number of edges $v_e$. This leads to a poor performance of algorithm \ref{alg:dyn}. To address this issue, we propose a linear-programming-based randomized algorithm \ref{alg:lprand}.

Focusing on the thin edges, let us modify the objective function in the optimization problem \eqref{opt:pdps} by merely including the savings from thin edges to formulate the problem \eqref{obj:pdpsthin} as follows:

\begin{subequations}\label{opt:pdpsthin}
\begin{align}\label{obj:pdpsthin}
\mbox{Max} \quad & z^{TN}_1 = \sum_{e\in E_{TN}} ( \sum_{(s,t):e\in E_{(s,t)}} x^{(s,t)}_e -y_e)c(e) \\ \label{conspdps:assignthin}
\mbox{s.t.} \quad & \sum_{\substack{ v_i: \\ e=(v_i,v)\in E(s,t)}} x_e^{(s,t)}- \sum_{\substack{ v_j: \\ e=(v,v_j)\in E(s,t)}} x_e^{(s,t)}=d_v^{(s,t)} && \forall v \in E, \forall (s,t) \in P
\\   \label{cons1:vectorthin}
& x_e^{(s,t)} \leq y_e   &&\scriptstyle{ \forall e \in E(G), \forall (s,t): e \in E_{(s,t)} }
\\  \label{conspdps:integthin}
& x_{e}^{(s,t)} \in \left \{0,1  \right \}, &&\substack{\forall (s,t)\in P, \\ \forall e \in E(s,t)}
\\  \label{conspdps:integauxthin}
& y_e \in \left \{0,1  \right \}, &&{\scriptstyle \forall e \in E}
\end{align}
\end{subequations}

Consider an optimal solution to problem \eqref{opt:pdpsthin} and denote it by $x^{TN}, y^{TN}$. Note that, if the instance $(G,P)$ is thin-dominant, the optimal objective value of problem \eqref{opt:pdpsthin} $z_1^{opt,TN}$ will be at least $(1-\frac{1}{m^{\epsilon}}) z_1^{opt}$, because $(x^{opt},y^{opt})$ is a feasible solution for problem \eqref{opt:pdpsthin} and also thin-dominance yields:
\begin{equation*}
    \sum_{e\in E_{TK}} ( \sum_{(s,t):e\in E_{(s,t)}} x^{opt,(s,t)}_e -y^{opt}_e)c(e)<\frac{1}{m^{\epsilon}} \sum_{e\in E} ( \sum_{(s,t):e\in E_{(s,t)}} x^{opt,(s,t)}_e -y^{opt}_e)c(e)
\end{equation*}

\begin{definition}\label{def:light}
Let us denote by $b_{(s,t)}$ the number of thin edges in the local graph $E_{(s,t)}$. We call a thin-dominant instance $(G,P)$ light if the portion of total savings \eqref{obj:pdpsthin} in the optimization problem \eqref{opt:pdpsthin} at $x^{TN}, y^{TN}$ in the absence of all demand pairs $(s,t)\in P$ with $b_{(s,t)} < \sqrt{m}$ decreases to less than $(1-\frac{1}{m^{\epsilon}})z_1^{opt,TN}$, where $z_1^{opt,TN}$ is the solution to the problem \eqref{opt:pdpsthin} for instance $(G,P)$, i.e.

\begin{equation*}
    \sum_{e\in E_{TN}} ( (\sum_{\substack{(s,t):e\in E_{(s,t)}\\ b_{(s,t)} > \sqrt{m}}} x^{TN,(s,t)}_e) - \max_{\substack{(s,t):e\in E_{(s,t)}\\ b_{(s,t)} > \sqrt{m}}} x^{TN,(s,t)}_e)c(e) \leq (1-\frac{1}{m^{\epsilon}}) \sum_{e\in E_{TN}} (( \sum_{(s,t):e\in E_{(s,t)}} x^{TN,(s,t)}_e) -y^{TN}_e)c(e)
\end{equation*}

Here, the $\max$ function over the empty set outputs $0$. We call a thin-dominant instance heavy, otherwise.
\end{definition}

We propose a linear relaxation of problem \eqref{opt:pdpsthin} in which we relax the binary constraints \eqref{conspdps:integ} and \eqref{conspdps:integaux}. We also restrict the objective function to the edges that are contained by at least one demand pair $(s,t)$ with $b_{(s,t)}< \sqrt{m}$. As such, we can provide the LP relaxation of the problem as:
\begin{subequations}\label{opt:pdpsrelax}
    \begin{align}\label{obj:pdpsrelax}
    \mbox{Max} \quad & z_2 =\sum_{\substack{e:\exists (i,j)| b_{(i,j)}\leq \sqrt{m} \\ e \in E_{(i,j)} \\ e \in E_{TN}} } ( \sum_{(s,t):e\in E_{(s,t)}} x^{(s,t)}_e -y_e)c(e) \\ \label{conspdps:assignrelax}
    \mbox{s.t.} \quad & \sum_{\substack{ v_i: \\ e=(v_i,v)\in E_{(s,t)}}} x_e^{(s,t)}- \sum_{\substack{ v_j: \\ e=(v,v_j)\in E_{(s,t)}}} x_e^{(s,t)}=d_v^{(s,t)} && \forall v_i \in V, \forall (s,t) \in P
    \\   \label{cons1:vectorrelax}
    & x_e^{(s,t)} \leq y_e   &&\scriptstyle{ \forall e \in E, \forall (s,t): e \in E_{(s,t)} }
    \end{align}
\end{subequations}

Consider an optimal solution to problem \eqref{opt:pdpsrelax} and denote it by $(x^*,y^*)$.  

\begin{lemma}\label{rem:discon}
If a thin-dominant instance $(G,P)$ is light, the optimal objective value in optimization problem \eqref{opt:pdpsrelax} is at least $\frac{1}{m^{\epsilon}} z^{TN}_1(x^{TN},y^{TN})$, In other words :
\begin{equation}\label{observ:minflow}
    \sum_{\substack{e:\exists (i,j)| b_{(i,j)}\leq \sqrt{m} \\ e \in E_{(i,j)} \\ e \in E_{TN}} }  ( (\sum_{(s,t):e\in E_{(s,t)}} x^{*,(s,t)}_e) -y^*_e)c(e) \geq \frac{1}{m^{\epsilon}} z_1^{TN}(x^{TN},y^{TN})
\end{equation}
\end{lemma}

\begin{proof}
See Appendix \ref{sec:app1}.
\end{proof}

\begin{lemma}\label{thm:main}
Given an optimal solution, $(x^*,y^*)$, to problem \eqref{opt:pdpsrelax}, using at most $O(|P| m\sqrt{m})$ operations, we can find another feasible solution $(\vec{x}^2,\vec{y}^2)$ fot problem \eqref{opt:pdpsrelax} whose savings on each thin edge $e$ that is contained by at least one of the local graphs $G_{(s,t)}$ where $b_{(s,t)} < \sqrt{m}$ satisfies the following two conditions:

\begin{equation*}
( \sum_{(s,t):e\in E_{(s,t)}} x^{2,(s,t)}_e -y^2_e)c(e) \geq \frac{1}{2} ( \sum_{(s,t):e\in E_{(s,t)}} x^{*,(s,t)}_e -y^*_e)c(e)
\end{equation*}
\begin{equation*}
    \max_{(s,t):e\in E_{(s,t)}}x^{2,(s,t)}_e \geq \frac{1}{2 m^{1/2}}
\end{equation*} 
\end{lemma}
\begin{proof}
See Section \ref{sec:appxalg}.
\end{proof}
Here, we present an LP-based randomized approximation algorithm for light thin dominant instances.

\vspace{0.5cm}
\begin{algorithm}[H]
    \caption{LP based randomized algorithm}
    \label{alg:lprand}
\vspace{0.5cm}
\begin{algorithmic}
 
 \STATE 1. Find the thin edges $e$ that are contained by at least one of the local graphs $G_{(s,t)}$ where $b_{(s,t)} < \sqrt{m}$ and form optimization problem \eqref{opt:pdpsrelax}.
 
 \STATE 2. Find the optimal solution $(x^*,y^*)$ for problem \eqref{opt:pdpsrelax}.
 \STATE 3. Modify the solution $(x^*,y^*)$ according to the proof of Lemma \eqref{thm:main} to find the feasible solution $(\vec{x}^2,\vec{y}^2)$.
 \STATE 4. Use Algorithm \ref{alg:transfer} in Appendix \ref{sec:app4} to find an equivalent path-based solution for link-based solution $(\vec{x}^2,\vec{y}^2)$ to the optimization problem \eqref{opt:pdpsrelax} and denote it by $\vec{f}$ (as a vector whose elements are the variables $f^{(s,t)}_p$, where $f^{(s,t)}_p$ is the flow on path $p$ for pair of nodes $(s,t)$).
 \STATE 5. For each pair of nodes $(s,t)$, pick one of the shortest paths $p$ from $s$ to $t$ with probability $f^{(s,t)}_p$, denote their union by subgraph $H$.
 \STATE 6. Output $H$ and $c(H)$.
\end{algorithmic}
\end{algorithm}
\vspace{0.5cm}

The last remaining case is the case of heavy thin-dominant instances. Instead of directly devising an algorithm for this case, we take advantage of an iterative algorithm that solves the general CSPDP problem as follows:

\vspace{0.5cm}
\begin{algorithm}[H]
    \caption{Main Algorithm}
    \label{alg:Main}
\vspace{0.5cm}
\begin{algorithmic}

 \STATE 1. Set LB=0 and $OPT=\emptyset$.
 \STATE 2. Input instance $(G,P)$, set $P^0=P$.
 \STATE 3. Apply Algorithm 1, 
 \STATE 4. If $c(H)\geq LB$, then $LB=c(H)$, $OPT=H$. 
 \STATE 5. Apply Algorithm 2
 \STATE 6. If $c(H)\geq LB$, then $LB=c(H)$, $OPT=H$.
 \STATE 7. Delete all pairs of nodes $(s,t) \in P$ with $b_{(s,t)}<\sqrt{m}$ to obtain a new instance of node pairs $P^1$.
 \STATE 8. Set $(G,P)=(G,P^1)$ and go to step 2.
 \STATE 9. Denote by $D$ the set of demand pairs $(s,t) \in P^0$ that $dist_{H}(s,t) > dist_{G}(s,t)$. For any demand pair $(s,t) \in D$, consider an arbitrary shortest path, $p_{(s,t)}$, from $s$ to $t$. Then set $H=H \cup_{(s,t) \in D} p_{(s,t)}$. 
 \STATE Output $H$ and $c(H)$.
\end{algorithmic}
\end{algorithm}
\vspace{0.5cm}

Here, the core idea is that if the instance $(G,P)$ is either a thick-dominant instance or a light thin-dominant instance, then we can find a good quality optimal solution at steps 3 and 4. Otherwise, we can find a new instance $(G,P^1)$ with significantly fewer number of demand pairs, and apply the same procedure on this new instance. However, as the number of demand pairs is at most the number of node pairs in graph $G$, this process will terminate fast.

\section{Algorithm Analysis}\label{sec:appxalg}

In this section, we analyze the performance of the algorithms developed in Section \ref{sec:algov}. Then, we examine their computational complexity. Specifically, Theorem \ref{thm:dynamic} provides a guarantee on the performance of Algorithm \ref{alg:dyn} for thick-dominant instances of CSPDP. Theorem \ref{lemma:minflow} quantifies the performance of Algorithm \ref{alg:lprand} on light thin-dominant instances. Finally, Theorem \ref{thm:final} demonstrates the approximation ratio for our main algorithm \ref{alg:Main}. Lemmas \ref{runt:dyn}, \ref{runt:lp} and \ref{runt:main} quantify the computational complexity of Algorithms \ref{alg:dyn}, \ref{alg:lprand} and \ref{alg:Main}, respectively.

\begin{theorem}\label{thm:dynamic}
If instance $(G,P)$ is a thick-dominant instance, then Algorithm \ref{alg:dyn} provides a feasible solution on subgraph $H$ with objective value of at least $\frac{z_1^{opt}}{m^{1/2 + 2\epsilon}}$, where $z_1^{opt}$ is the optimal value of the problem \eqref{opt:pdps}.
\end{theorem}
    
\begin{proof}
   
Consider a thick-dominant optimal solution $(x^{opt},y^{opt})$ for problem \eqref{opt:pdps}.  Thick-dominance yields:
    
    \begin{equation}\label{eq:dyncount}
        \frac{|P|}{m^{1/2+2\epsilon}} z_1^{opt} \leq \frac{|P|}{m^{1/2+\epsilon}}\sum_{e\in E_{TK}} ( (\sum_{(s,t):e\in E_{(s,t)}} x^{opt,(s,t)}_e) -y^{opt}_e)c(e)
    \end{equation}
On the other hand, for each thick edge from definition \ref{def:thickedge} we have $\frac{|P|}{m^{1/2+\epsilon}} \leq v_e$. As such:  
    
    \begin{equation}\label{eq:dcountthick}
        \frac{|P|}{m^{1/2+\epsilon}}\sum_{e\in E_{TK}} ( (\sum_{(s,t):e\in E_{(s,t)}} x^{opt,(s,t)}_e) -y^{opt}_e)c(e)\leq \sum_{e\in E_{TK}} ( (\sum_{(s,t):e\in E_{(s,t)}} x^{opt,(s,t)}_e) -y^{opt}_e)v_e c(e) 
    \end{equation}
    
Here, we use Lemma \ref{lemma:updc} to upper bound the expression in \eqref{eq:dcountthick}.
\begin{lemma}\label{lemma:updc}
Any optimal solution $(\vec{x^{opt}},\vec{y^{opt}})$ for problem \ref{opt:pdps} satisfies the following inequality:
\begin{equation}\label{dynmatch}
    ((\sum_{(s,t):e\in E_{(s,t)}} x^{opt,(s,t)}_e) -y^{opt}_e)v_l c(e) \leq ( \sum_{(s,t):e\in E_{(s,t)}} x^{opt,(s,t)}_e)(v_e-1)c(e)    
\end{equation}

\end{lemma}

\begin{proof}
See Appendix \ref{sec:app2}.
\end{proof}

Now, let us change the order of summation in expression \eqref{dynmatch} to obtain:   
    \begin{equation*}
        \sum_{e\in E_{TK}} [( \sum_{(s,t):e\in E_{(s,t)}} x^{opt,(s,t)}_e)(v_e-1)c(e)]=\sum_{(s,t) \in P} ( \sum_{\substack{e\in E_{(s,t)} \\ e\in E_{TK}}} x^{opt,(s,t)}_e (v_e-1)c(e)) 
    \end{equation*}
    
    \begin{equation}\label{eq:dynedge}
        \leq |P| \max_{(s,t)}(\sum_{\substack{l\in E_{(s,t)} \\ e\in E_{TK}}} x^{opt,(s,t)}_e (v_e-1)c(e))= |P|a_{max}
    \end{equation}
    Where the last equality holds since $a_{max}$, obtained in Algorithm \ref{alg:dyn}, equals the weight of maximum weighted path between pair of nodes $(s,t) \in P$ in graph $H_{(s,t)}$.

    Now, we can combine \eqref{eq:dyncount}, \eqref{eq:dcountthick}, \eqref{dynmatch}, and \eqref{eq:dynedge} to prove:
    \begin{equation*}
    \frac{|P|}{m^{1/2+2\epsilon}} z_1^{opt} \leq|P|a_{max} \rightarrow     \frac{1}{m^{1/2+2\epsilon}} z_1^{opt} \leq a_{max}
    \end{equation*}
    
On the other hand, the overlap of any subgraph $H^{*}_{(s,t)}$ with path $p_{(s^*,t^*)}$ is a connected subpath of $p_{(s^*,t^*)}$ which is contained by path $h_{(s,t)}$. As a result, any edge $e \in H$ is contained by at least $v_e$ paths $h_{(s,t)} \: (s,t) \in P$. As a result, the savings in problem \ref{opt:pdps} for feasible solution $H$ is at least $a_{max}$.

\end{proof}

Proof of Lemma \ref{thm:main}:

\begin{proof}

For each pair of nodes $(s,t) \in P$ with $b(s,t)<\sqrt{m}$ and each thin edge $e \in G_{(s,t)}$, we can extend edge $e$ to a shortest path $p^e_{(s,t)}$ from $s$ to $t$. To do so, just note that $G_{(s,t)}$ is a dag with a source node $s$ and a sink node $t$, so each directed in $G_{(s,t)}$ starting from $e$ ends at node $t$ and each reverse directed path starting from $e$ ends $s$. Note that this requires at most $|E|$ operations, since any path in $G_{(s,t)}$ consists of at most $|E|$ edges. As we have $b(s,t)<\sqrt{m}$, number of such paths is at most $\sqrt{m}$. As a result, applying the same procedure on all local graphs $G_{(s,t)}$ takes at most $|P|m\sqrt{m}$ operations. Also, note that the union of paths $p^e_{(s,t)}$ includes at most $|P|m\sqrt{m}$ edges. Set $f^{1,(s,t)}_{p^e_{(s,t)}}=\frac{1}{|b_{(s,t)}|}$.
Moreover, for each pair of nodes $(s,t) \in P$ with $b(s,t)\geq \sqrt{m}$ consider an arbitrary shortest path $p_{(s,t)}$. Also, note that the union of paths $p^e_{(s,t)}$ includes at most $|P|m$ edges. Set $f^{1,(s,t)}_{p_{(s,t)}}=1$.


Now, consider the unique link-based solution $\vec{x}^1$, (as a vector whose elements are the variables $x^{(s,t)}_e$), and the corresponding capacity values $y^{1}_e=\max_{(s,t):e\in E_{(s,t)}} \vec{x}^{1,(s,t)}_e$ corresponding to $\vec{f}^1$, (as a vector whose elements are the variables $f^{(s,t)}_p$).

To obtain $x^{(s,t)}_e$ from $f_1$, consider an initial solution $\vec{x}=\vec{0}$. Then, for each edge $e$ in the union of paths $(s,t)$ with $b(s,t)<\sqrt{m}$, add $x^{(s,t)}_e$ by $f^{1,(s,t)}_{p^e_{(s,t)}}$, and for paths $p_{(s,t)}$ for demand pairs $(s,t)$ with $b(s,t)\geq \sqrt{m}$ add $x^{(s,t)}_e$ by $f^{1,(s,t)}_{p_{(s,t)}}$.

As we have at most $2|P| m\sqrt{m}$ edges in the union of all paths mentioned above, we can obtain $x^{(s,t)}_e$ from $f_1$ by at most $2|P| m\sqrt{m}$ operations.

We prove that the solution $\vec{x}^2=\frac{1}{2}(\vec{x}^1+\vec{x}^*)$, $y^{2,(s,t)}_e= \max_{(s,t):e\in E_{(s,t)}} \frac{1}{2}( x^{1,(s,t)}_e+x^{*,(s,t)}_e)$ satisfies the requirement. To do so, note that for each thin edge $e$ that is contained in the local graph of a pair $(i,j) \in P$ such that $b_{(i,j)}\leq \sqrt{m}$ we have $\max_{(s,t):e\in E_{(s,t)}}x^{1,(s,t)}_e \geq \frac{1}{2}\max_{(s,t):e\in E_{(s,t)}}x^{1,(s,t)}_e\geq \frac{1}{2} x^{1,(i,j)}_e \geq \frac{1}{2\sqrt{m}}$.

Moreover, we can distribute the max in the definition of $y^{2,(s,t)}_e$ to conclude:

\begin{equation*}
    \max_{(s,t):e\in E_{(s,t)}} \frac{1}{2}( x^{1,(s,t)}_e+x^{*,(s,t)}_e) \geq \frac{1}{2}(\max_{(s,t):e\in E_{(s,t)}} ( x^{1,(s,t)}_e)+ \max_{(s,t):e\in E_{(s,t)}} (x^{*,(s,t)}_e))
\end{equation*}

This results in:

\begin{equation*}
    ( \sum_{(s,t):e\in E_{(s,t)}} x^{2,(s,t)}_e -y^2_e)c(e) \geq \frac{1}{2}( \sum_{(s,t):e\in E_{(s,t)}} x^{*,(s,t)}_e  -y^*_e)c(e) + \frac{1}{2}( \sum_{(s,t):e\in E_{(s,t)}} x^{1,(s,t)}_e -y^1_e)c(e)
\end{equation*}

Given $( \sum_{(s,t):e\in E_{(s,t)}} x^{1,(s,t)}_e -y^1_e)c(e) \geq 0$, we conclude:

\begin{equation*}
    ( \sum_{(s,t):e\in E_{(s,t)}} x^{2,(s,t)}_e -y^2_e)c(e) \geq \frac{1}{2}( \sum_{(s,t):e\in E_{(s,t)}} x^{*,(s,t)}_e  -y^*_e)c(e) 
\end{equation*}

 This concludes the proof.

\end{proof}

\begin{theorem}\label{lemma:minflow}
Given a light thin-dominant instance $(G,P)$ with the modified solution $(\vec{x}^2,\vec{y}^2)$ at step 2 of Algorithm \ref{alg:lprand}, steps 3 and 4 provide a solution for problem \eqref{opt:pdps} with objective value \eqref{obj:pdps} being at least $(1-\frac{1}{m^{\epsilon}}) \frac{1}{4m^{1/2+2\epsilon}} z_1(x^{opt},y^{opt})$.
\end{theorem}

\begin{proof}
For any link-based solution, $(\vec{x}^2,\vec{y}^2)$, there exists an equivalent path-based solution $\vec{f}^2$ where $f^{2,(s,t)}_p$ is the flow on path $p$ for pair of nodes $(s,t)$. Let us assume for each pair of nodes $(s,t)$ we pick one of the shortest paths $p \in P_{(s,t)}$ with probability $f^{(s,t)}_p$. Denote by $s_e$ and $s$ the part of the objective function \eqref{obj:pdpsrelax} associated with thin edge $e$ that is contained in the local graph of a pair $(i,j) \in P$ such that $b_{(i,j)}\leq \sqrt{m}$, and the total objective value \eqref{obj:pdpsrelax}, respectively. We can compute the expected values for $s_e$ and $s$ as follows:

\begin{equation*}
    s_e=( \sum_{(s,t):e\in E_{(s,t)}} x^{2,(s,t)}_e -y^2_e)c(e)=( \sum_{(s,t):e\in E_{(s,t)}} x^{2,(s,t)}_e -\max_{(s,t):e\in E_{(s,t)}}x^{2,(s,t)}_e)c(e)
\end{equation*}

Now, we compute the average savings in a link $e \in E$ as follows:

\begin{equation*}
   \mathbb{E}(s_e)= \mathbb{E}[( (\sum_{(s,t):e\in E_{(s,t)}} x^{2,(s,t)}_e) -y^2_e)c(e)]= ((\sum_{(s,t):e\in E_{(s,t)}} x^{2,(s,t)}_e)-(1-\prod_{(s,t):e\in E_{(s,t)}} (1-x^{2,(s,t)}_e)))c(e)
\end{equation*}
Denote by $x^{2,(i,j)}_e=\max_{(s,t):e\in E_{(s,t)}}x^{2,(s,t)}_e$. Then, we obtain:
\begin{equation}\label{eq:2part}
   \mathbb{E}(s_e) = ((\sum_{\substack{(s,t):e\in E_{(s,t)}\\ (s,t)\neq (i,j)}} x^{2,(s,t)}_e)-(1-\prod_{\substack{(s,t):e\in E_{(s,t)}\\(s,t)\neq (i,j)}} (1-x^{2,(s,t)}_e)) +(x^{2,(i,j)}_e- x^{2,(i,j)}_e \prod_{\substack{(s,t):e\in E_{(s,t)}\\(s,t)\neq (i,j)}} (1-x^{2,(s,t)}_e)))c(e)
\end{equation}
The right-hand side of Equation \eqref{eq:2part} can be rewritten as:

\begin{equation*}
  x^{2,(i,j)}_e(1- \prod_{\substack{(s,t):e\in E_{(s,t)}\\(s,t)\neq (i,j)}} (1-x^{2,(s,t)}_e))c(e)
\end{equation*}
\begin{equation}\label{eq:partright}
  =  x^{2,(i,j)}_e((1- \prod_{\substack{(s,t):e\in E_{(s,t)}\\(s,t)\neq (i,j)}} (1-x^{2,(s,t)}_e)-\sum_{\substack{(s,t):e\in E_{(s,t)}\\ (s,t)\neq (i,j)}} x^{2,(s,t)}_e))c(e)+x^{2,(i,j)}_e(\sum_{\substack{(s,t):e\in E_{(s,t)}\\ (s,t)\neq (i,j)}} x^{2,(s,t)}_e)c(e)
\end{equation}

We can use \eqref{eq:partright} to rewrite expression \eqref{eq:2part} as follows:

\begin{equation}\label{eq:compact}
    [((\sum_{\substack{(s,t):e\in E_{(s,t)}\\ (s,t)\neq (i,j)}} x^{2,(s,t)}_e)-1+\prod_{\substack{(s,t):e\in E_{(s,t)}\\(s,t)\neq (i,j)}} (1-x^{2,(s,t)}_e)) (1-x^{2,(i,j)}_e)+x^{2,(i,j)}_e(\sum_{\substack{(s,t):e\in E_{(s,t)}\\ (s,t)\neq (i,j)}} x^{2,(s,t)}_e)]c(e)
\end{equation}

\begin{lemma}\label{eq:partleft}
For $x^{2,(s,t)}_e \geq 0 : \forall (s,t) \in P$ we have: 
\begin{equation}
    \sum_{\substack{(s,t):e\in E_{(s,t)}\\ (s,t)\neq (i,j)}} x^{2,(s,t)}_e-1+\prod_{\substack{(s,t):e\in E_{(s,t)}\\(s,t)\neq (i,j)}} (1-x^{2,(s,t)}_e)\geq 0
\end{equation}
\end{lemma}

\begin{proof}
See Appendix \ref{sec:app3}.
\end{proof}

We can use Lemma \ref{eq:partleft} to conclude that the left-hand side of Equation \eqref{eq:compact} is positive. Hence:
\begin{equation*}
    \mathbb{E}(s_e)\geq x^{2,(i,j)}_e (\sum_{\substack{(s,t):e\in E_{(s,t)}\\ (s,t)\neq (i,j)}} x^{2,(s,t)}_e)c(e) \geq \frac{1}{2m^{1/2}} s_e
\end{equation*}

Using the same argument for all thin edges $e \in E$ we conclude: 
\begin{equation*}
    \mathbb{E}[S]\geq \sum_{\substack{e:\exists (i,j)| b_{(i,j)}\leq \sqrt{m} \\ e \in E_{(i,j)} \\ e \in E_{TN}} } \frac{1}{2\sqrt{m}} ( \sum_{(s,t):e\in E_{(s,t)}} x^{2,(s,t)}_e -y^2_e)c(e)
\end{equation*}  

As the solution $(\vec{x}^2,\vec{y}^2)$ satisfies the condition in Lemma \ref{thm:main} we conclude:
\begin{equation*}
    \sum_{\substack{e:\exists (i,j)| b_{(i,j)}\leq \sqrt{m} \\ e \in E_{(i,j)} \\ e \in E_{TN}} } \frac{1}{2\sqrt{m}} ( \sum_{(s,t):e\in E_{(s,t)}} x^{2,(s,t)}_e -y^2_e)c(e) \geq \sum_{\substack{e:\exists (i,j)| b_{(i,j)}\leq \sqrt{m} \\ e \in E_{(i,j)} \\ e \in E_{TN}} } \frac{1}{4\sqrt{m}} ( \sum_{(s,t):e\in E_{(s,t)}} x^{*,(s,t)}_e -y^*_e)c(e)
\end{equation*}

As instance $(G,P)$ is a light thin-dominant instance, when we restrict the objective function \eqref{obj:pdpsrelax} into the thin edges that are contained by at least one demand pair $(s,t)$ with $b_{(s,t)}< \sqrt{m}$, the objective has to be at least $\frac{1}{m^{\epsilon}} z_2(x^*,y^*)$, As such :
\begin{equation}
    \sum_{\substack{e:\exists (i,j)| b_{(i,j)}\leq \sqrt{m} \\ e \in E_{(i,j)} \\ e \in E_{TN}} }  ( (\sum_{(s,t):e\in E_{(s,t)}} x^{*,(s,t)}_e) -y^*_e)c(e) \geq \frac{1}{m^{\epsilon}} z_2(x^*,y^*)
\end{equation}

We multiply both sides by $\frac{1}{4 \sqrt{m}}$ to conclude:

\begin{equation*}
 \sum_{\substack{e:\exists (i,j)| b_{(i,j)}\leq \sqrt{m} \\ e \in E_{(i,j)} \\ e \in E_{TN}} } \frac{1}{4\sqrt{m}} ( \sum_{(s,t):e\in E_{(s,t)}} x^{*,(s,t)}_e -y^*_e)c(e) \geq  \frac{1}{4m^{1/2+\epsilon}} z_2(x^*,y^*)
\end{equation*}

Now, we can use Remark \ref{rem:discon}

\begin{equation*}
  \frac{1}{4m^{1/2+\epsilon}} z_2(x^*,y^*)  \geq      \frac{1}{4m^{1/2+\epsilon}} \frac{1}{m^{\epsilon}} z^{TN}_1(x^{TN},y^{TN})
\end{equation*}

On the other hand, for a thin dominant instance $(G,P)$ we have that $z^{opt,TN}_1$ is at least $(1-\frac{1}{m^{\epsilon}}) z_1^{opt}$. As such,
\begin{equation*}
    \mathbb{E}[S] \geq (1-\frac{1}{m^{\epsilon}}) \frac{1}{4m^{1/2+2\epsilon}}  z_1(x^{opt},y^{opt})
\end{equation*}

This concludes the proof.

\end{proof}

Proof of Theorem \ref{thm:final}:
    
\begin{proof}
We prove Algorithm \ref{alg:Main} finds a feasible solution for Problem \eqref{opt:pdps} with objective value of at least $\frac{1}{m^{1/2+\epsilon}} z_1^{opt}$. First, note that if instance $(G,P)$ is either thick-dominant or light thin-dominant, then by Theorems \ref{thm:dynamic} and \ref{lemma:minflow} we find a feasible solution for Problem \eqref{opt:pdps} with objective value of at least $(1-\frac{1}{m^{\epsilon}})\frac{1}{4m^{1/2+2\epsilon}} z_1^{opt}$ in steps 3 and 5 of Algorithm \ref{alg:Main}. Otherwise, We can remove all pairs of nodes $(s,t) \in P$ that satisfy $b_{(s,t)}<\sqrt{m}$ to obtain a new instance of node pairs $P^1$. As $(G,P)$ is a thin-dominant instance, $z_1^{opt,TN}$, the optimal value of the optimization \eqref{opt:pdpsthin}, is at least $(1-\frac{1}{m^{\epsilon}})z_1^{opt}$. As $(G,P)$ is a heavy instance, using Definition \ref{def:light}, the solution $x^{TN,(s,t)}_e , y^{TN}_e= \max_{\substack{(s,t)\\ b_{(s,t)} > \sqrt{m}}} x^{TN,(s,t)}_e$ is a feasible solution for instance $(G,P^1)$ with objective value of at least $(1-\frac{1}{m^{\epsilon}})z_1^{opt,TN}$. As a result, the optimal value for Problem \ref{opt:pdps} with instance $(G,P^1)$ is at least $(1-\frac{1}{m^{\epsilon}})^2 z_1^{opt}$.

On the other hand, in this case, each pair of nodes in $P^1$ has at least $\sqrt{m}$ thin edges. Moreover, each thin edge $e$ is contained in at most $\frac{|P|}{m^{1/2+\epsilon}}$ local graphs $G_{(s,t)}$. Using the Pigeonhole principle, we should have at least  $\frac{|P^1|\sqrt{m}}{|P|/m^{1/2+\epsilon}}=\frac{|P^1|}{|P|}m^{1+\epsilon}$ thin edges, which should be less than total number of edges $m$. 

As a result we have:
\begin{equation*}
   \frac{|P^1|}{|P|}<m^{-\epsilon} \rightarrow  |P^1|< \frac{|P|}{m^{\epsilon}} 
\end{equation*}

Note that $\frac{1}{m^{\epsilon}}^{2/\epsilon}=\frac{1}{m^2}$. Moreover, note that the set of all demand pairs in the original instance $P$ to Algorithm \ref{alg:Main} has a size less than the total number of pairs of vertices $\binom{|V|}{2}$, which is bounded above by $\leq m^2$ in a connected graph $G$. As a result, we shrink the set of demand pairs in step 6 of Algorithm \ref{alg:Main} by at most $\frac{2}{\epsilon}$ times. 

Denote by $(G,P')$ the instance obtained from the $i$-th visit to step 6 in an iteration of Algorithm \ref{alg:Main}. If $(G,P')$ is a thick-dominant or light thin-dominant, then by Theorems \ref{thm:dynamic} and \ref{lemma:minflow} we find a feasible solution for Problem \eqref{opt:pdps} with the objective value of at least $(1-\frac{1}{m^{\epsilon}})^{2i+1}\frac{1}{4m^{1/2+2\epsilon}} z_1^{opt}$. As the shrinking procedure in step 6 of Algorithm \ref{alg:Main} cannot be applied for more than $\frac{2}{\epsilon}$ times, we are able to find a feasible solution for Problem \eqref{opt:pdps} with objective value of at least
\begin{equation*}
    (1-\frac{1}{m^{\epsilon}})^{(2/\epsilon)+1}\frac{1}{4m^{1/2+2\epsilon}} z_1^{opt}
\end{equation*}

However, for a given $0 < \epsilon <1$ the Bernoulli inequality yields: 
\begin{equation*}
    \lim_{m\rightarrow \infty}(1-\frac{1}{m^{\epsilon}})^{(2/\epsilon)+1} \geq  \lim_{m\rightarrow \infty} 1-\frac{2+\epsilon}{\epsilon m^{\epsilon}} = 1
\end{equation*}

This concludes the proof.  

\end{proof}

\begin{lemma}\label{runt:dyn}
Algorithm \ref{alg:dyn} runs in $O(|P||E|^2)$ operations.
\end{lemma}

\begin{proof}
Running standrd algorithms we can construct the shortest path subgraphs $G_{(s,t)}$ in $O(m \log m)$ operations. As such, in step 1 we need $O(|P|m \log m)$ operations. We can find $v_e$ the number of local graphs $G_{(s,t)}$ containing edge $e$ by simply checking all the edges in all the subgraphs $G_{(s,t)}$, we check at most $O(|P||E|)$ edges, so step 2 requires at most $O(|P||E|^2)$ operations.

To find the maximum weighted paths in graphs $H_{(s,t)}$ and $H^*_{(s,t)}$, note that if two vertices $x$ and $y$ in a local graph $H_{(s,t)}$ are connected by two zero length oppositely directed edges, $c(x,y)=c(y,x)=0$, we can merge the two vertices by contracting the edge between them \citep{west1996introduction}. Doing so, the length of the longest path doesn't change. As a result, finding the maimum weighted path in graphs $H_{(s,t)}$ and $H^*_{(s,t)}$ is equivalent to finding the maimum weighted path in Directed Acyclic Graphs that form by merging vertices that are connected through double sided zero length directed edges. This can be done in $O(|V|+|E|)$ operations using the idea of topological sorting. Hence, steps 3 and 7 require at most $O(|P||E|)$ operations. In steps 5 and 6, for each demand pair $(s,t) \in P$ we have to consider the intersection of edges in $H^*_{(s,t)}$ and $p_{(s^*,t^*)}$, this requires at most $|E|$ number of operations. Hence steps 5 and 6 take at most $|P||E|$ number of operations. Note that steps 4 and 8 also will not take more than $O(|P|)$ operations. This concludes the proof. 
\end{proof}

\begin{lemma}\label{runt:lp}
Algorithm \ref{alg:lprand} runs in $O((|P||E|)^{2.38})$ operations. 
\end{lemma} 
    
\begin{proof}

Similar to Lemma \ref{runt:dyn} we find $v_e$ for all edges $e \in E$ by at most $O(|P||E|^2)$ operations. Then we can find thin and thick edges with $O(E)$ operations. Afterwards, we find the demand pairs $(s,t) \in P$ with $b_{(s,t)}<\sqrt{m}$ by counting all thin edges of all subgraphs $G_{(s,t)}$ which requires at most $O(|P||E|)$ operations. Then we take the union of all thin edges in such subgraphs to obtain required set of edges to form optimization problem \ref{opt:pdpsrelax}. This takes at most $O(|P||E|)$ operations. Hence, step 1 requires at most $O(|P||E|^2)$ operations.
In step 2 we are required to solve a linear program with $|P||E|$ variables and $O(|P||E|)$ constraints. This requires at most $O((|P||E|)^{\omega})$ operations, where $\omega \approx 2.38$ is the exponent of matrix multiplication, using the novel algorithm of \cite{cohen2019solving} for solving linear programming problems. We apply Lemma \ref{thm:main} to conclude step 3 takes at most $O(|P|m \sqrt{m})$ operations.

In step 4, we use Algorithm \ref{alg:transfer}. This algorithm will terminate in at most $|P||E|$ iterations, since in each iteration at least for one demand pair $(s,t)$ and one edge $e$ the value of $x^{2,(s,t)}_g$ will be substituted by 0. Also, similar to Lemma \ref{thm:main} each iteration requires at most $|P||E|$ operations.
Finally, step 5 requires $|P|$ operations.

\end{proof}    
    
\begin{lemma}\label{runt:main}
Algorithm \ref{alg:Main} runs in $O((|P||E|)^{2.37} (1/\epsilon))$ operations.

\end{lemma} 
        
\begin{proof}
The proof directly follows from the proof of Lemmas \ref{runt:dyn} and \ref{runt:lp} and the fact that the Algorithm \ref{alg:Main} has at most $\frac{2}{\epsilon}$ iterations.
\end{proof}

\section{Hardness}\label{sec:hardness}
\noindent Here, we represent our reduction from the MAX REP instance of the Label cover$_{\max}$ problem to the directed CSPDP, the reduction to undirected CSPDP follows similarly. LABEL COVER was first introduced in \cite{arora1997hardness}. This problem has shown strong hardness results for many NP-hard problems. It is known that for LABEL COVER, there is no approximation algorithm achieving a ratio $2^{\log^{1-\epsilon}n}$, for
any $0< \epsilon< 1$, unless $NP \subseteq DTIME(n^{polylog(n)})$ \citep{arora1997hardness,hochba1997approximation}. The best approximation algorithm for label cover is the one introduced in \cite{charikar2009improved}, where they propose a $O(n^{1/3})$)-approximation algorithm for MAX REP.

In Max Rep we are given a bipartite graph $G(V_1, V_2, |E|)$. The sets $V_1$ and $V_2$ are split into a disjoint union of $k$ sets: $V_1 = \cup_{i=1}^{k} A_i$ and $V_2 = \cup_{j=1}^{k} B_j$. The sets $A_i$ and $B_j$ all have size $\frac{n}{k}$. The bipartite graph and the partitions of $V_1$ and $V_2$ induce a super-graph $H$ in the following way: The vertices in $H$ are the sets $A_i$ and $B_j$. Two sets $A_i$ and $B_j$ are connected by a (super) edge in $H$ iff there exist $a_i \in A_i$ and $b_j \in B_j$ , which are adjacent in $G$.

We must select a single ``representative'' vertex $a_i \in A_i$ from each subset $A_i$ , and a single “representative” vertex
$b_j \in B_j$ from each $B_j$. We consider that a super-edge ($A_i, B_j$) is covered if the two corresponding representatives are neighbors in G, i.e., $(a_i, b_j) \in E$. The goal is to select a single representative from each set and maximize the number of super-edges covered. We define $u_{max}=3\:d_{max}$ where $d_{max}$ is the maximum degree of all the vertices in $V_2$. Also, for any vertex $b_i \in V_2$ we define $u_{b_i}=3\:d_{v_i}-1$, where $d_{v_i}$ is the degree of vertex $v_i$.

Proof of Theorem \ref{thm:hardnessfinal}:

\begin{proof}

Consider an instance of MAX REP $G(V_1, V_2, E)$ with partitions $V_1=\bigcup_{i=1}^k A_i$, and $V_2=\bigcup_{i=1}^k B_i$. The sets $A_i$ and $B_j$ all have size $\frac{n}{k}$. We construct an instance of cost sharing pairwise distance preservers $(H,P)$ with $|P|=2k$ demand pairs and $O(n^2)$ edges. Any feasible solution for CSPDP instance $(H,P)$ corresponds to at least one feasible solution for MAX REP instance $G(V_1, V_2, E)$ with the same objective value, while any feasible solution for MAX REP instance $G(V_1, V_2, E)$ corresponds to exactly one feasible solution for CSPDP instance $(H,P)$ with the same objective value. 

Here, we present a reduction from directed CSPDP to MAX REP by constructing an instance $(H,P)$ for directed CSPDP. If we consider the undirected counterpart of the same graph $H$ with the same set of demand pair $|P|$, we, similarly, obtain a reduction from MAX REP to the undirected CSPDP.

For each set $A_i: i \in \left \{ 1,\dots,k \right \}$, we define a weighted graph $H^1_i$ that consists of a source node $s_i$ and a sink node $t_i$ and the union of $\frac{n}{k}$ disjoint paths, $p_v: v \in A_i$, between $s_i$ and $t_i$. Each path, $p_v= (s_i, x_v^1, x_v^2, \dots, x_v^{2n} , t_i)$, corresponds to one of the vertices $v \in A_i$ and consists of $2n+1$ edges. We assign weight $u_{max}$ to edges $(s_i, x_v^1)$, $(x_v^{2n} , t_i)$ and all the edges $(x_v^{i}, x_v^{i+1}) : i \equiv 0 (mod 2) $, and assign weight 1 to The edges $(x_v^{i}, x_v^{i+1}) : i \equiv 0 (mod 2) $. Figure \ref{fig:gi} demonstrates construction of the graph $H_i^1$.

\begin{figure}
    \centering
    \includegraphics[width=0.5\linewidth]{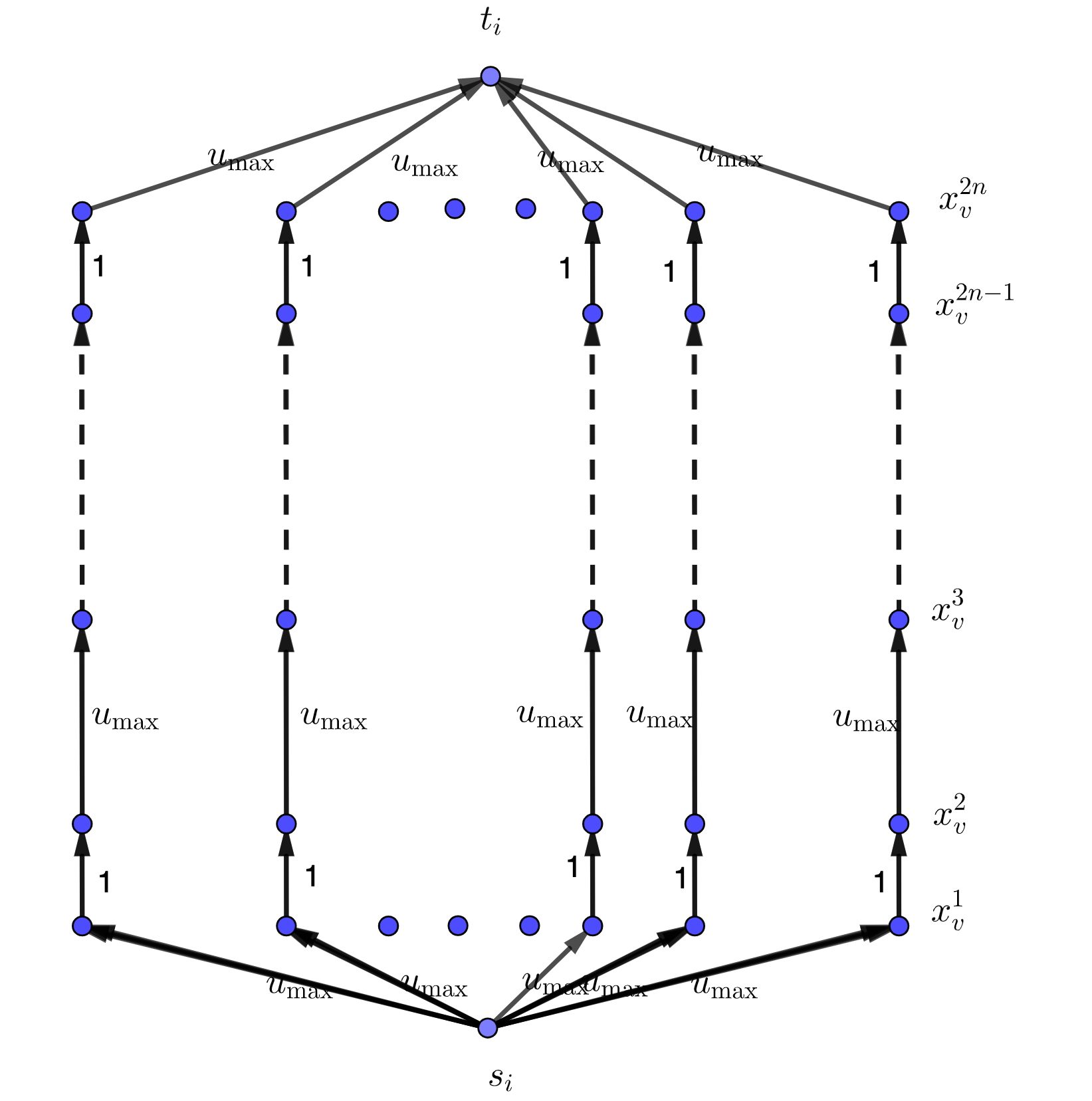}%
    \caption{construction of graph $H_i^1$}
    \label{fig:gi}
\end{figure}

For each set $B_j: i \in \left \{ 1,\dots,k \right \}$, we define a weighted graph $H^2_j$ that consists of a source node $o_j$ and a sink node $d_j$ and a union of $\frac{n}{k}$ edges coming out of $o_j$. Each edge $q_u= (o_i, y_u)$ represents one of the vertices $u \in B_i$. The edges coming out of $o_i$ have weight $2n u_{max}$. Figure \ref{fig:hi} demonstrates the construction of the graph $H_j^2$. We will connect the vertices $y_u$ to the sink nodes $d_i$ by adding some edges to the graph formed by the union $\bigcup_i H^1_i \bigcup_i H^2_i$ in next step.

\begin{figure}
    \centering
    \includegraphics[width=0.5\linewidth]{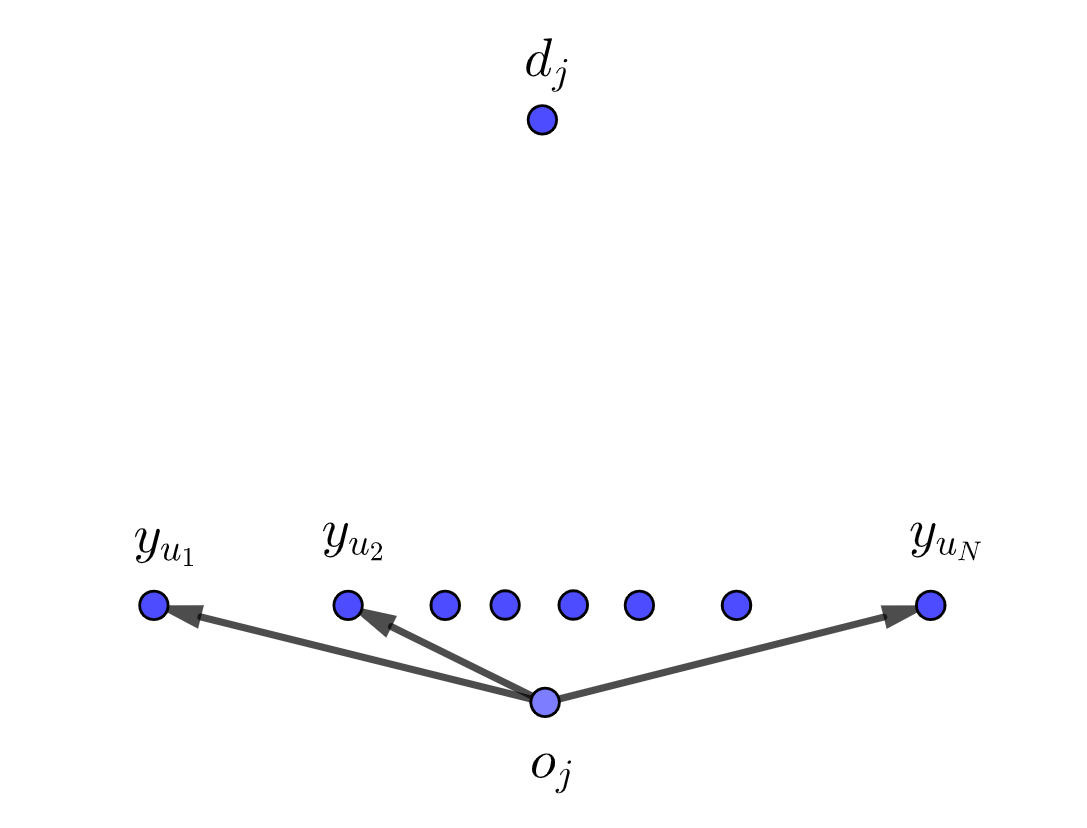}%
    \caption{construction of graph $H_i^2$}
    \label{fig:hi}
\end{figure}

Now, we consider the union $\bigcup_i H^1_i \bigcup_i H^2_i$ and we add edges to construct the graph $H$ for the CSPDP instance as follows. For each vertex $u_r \in B_i: r\in \left \{ 1,2,\dots, 2n \right \}$, we add edges to the graph to construct a new path between $o_i$ and $d_i$ as follows. First, consider all the vertices $v_{i_1}, v_{i_2}, \dots, v_{i_{l_u}} \subseteq V_1$ from which there exists an edge to $u_r$ in $G$. Then, add an edge from $y_{u_r}$ to $x_{v_1}^{2r-1}$ with weight 0. Then, for each $i \in \left \{ 1,2,\dots, l_u-1 \right \}$ add an edge between $x_{v_{i}}^{2r}$ and $x_{v_{i+1}}^{2r-1}$ with weight $2$. Then, add an edge from $x_{v_{l_u}}^{2r}$ to $d_i$ with weight $(2n+1) u_{max}- u_{b_i}$. Figure \ref{fig:H} demonstrates the construction of the graph $H$.

\begin{figure}
    \centering
    \includegraphics[width=1\linewidth]{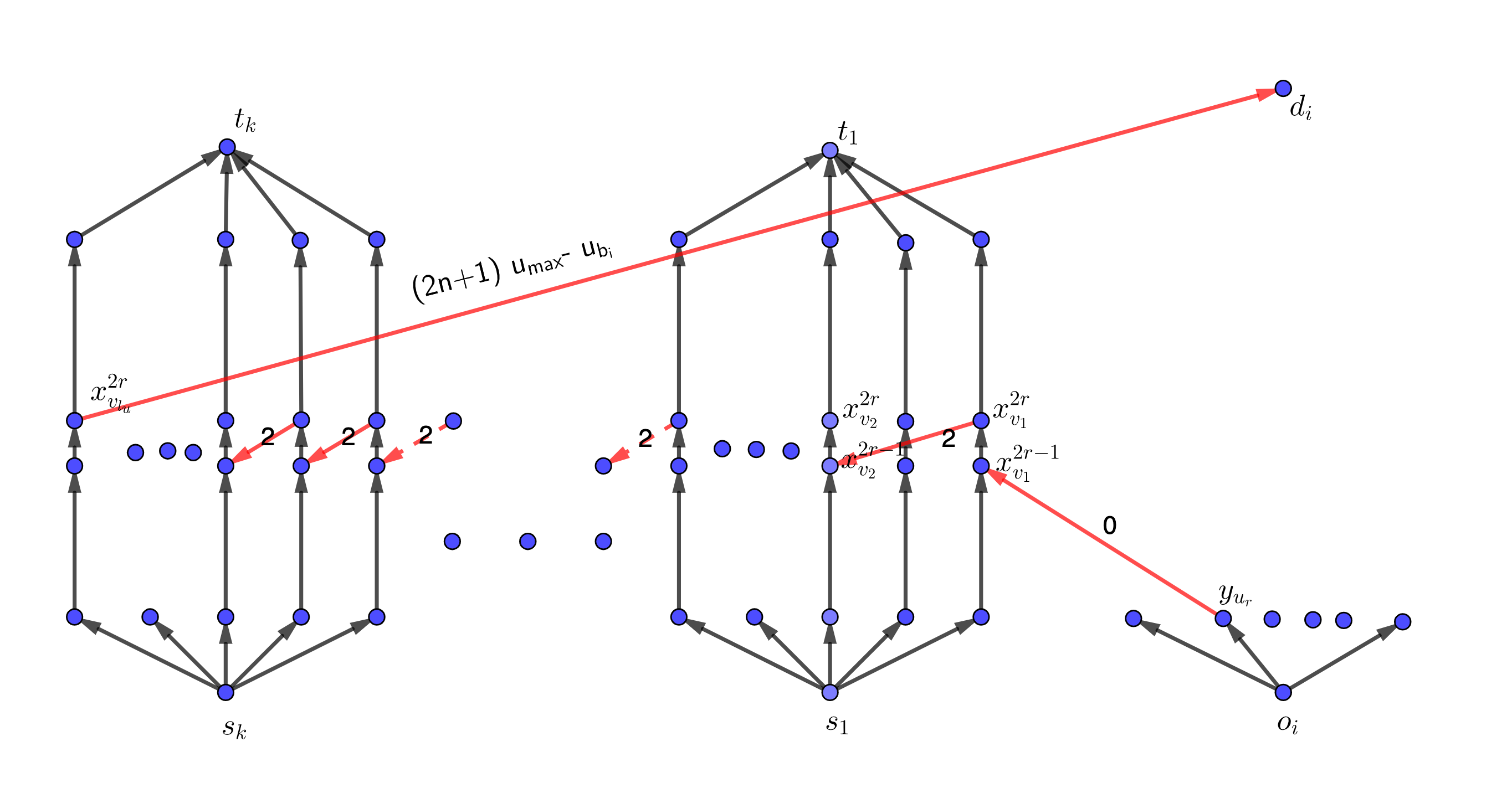}%
    \caption{construction of graph $H$}
    \label{fig:H}
\end{figure}

Now, consider the CSPDP instance with the underlying graph $H$, and the set of demand pairs $\bigcup_{i} (s_i,t_i) \bigcup_{i} (o_i,d_i)$.

\begin{lemma}\label{lemma:shph1}
The only shortest paths from $s_i$ to $t_i$ in graph $H$ are the paths $p_v: v \in A_i$
\end{lemma}
\begin{proof}
See Appendix \ref{sec:app5}.
\end{proof}    

\begin{lemma}\label{lemma:shph2}
The only shortest paths from $o_i$ to $d_i$ in graph $H$ are 

the paths 
$q_u= (o_i,y_{u_r}, x_{v_1}^{2r-1}, x_{v_1}^{2r},\dots ,x_{v_i}^{2r-1}, x_{v_i}^{2r}, \dots, x_{v_{l_u}}^{2r-1}, x_{v_{l_u}}^{2r}, d_i)$
\end{lemma} 

\begin{proof}
See Appendix \ref{sec:app6}.
\end{proof}

Given an optimal solution $H^*$ to CSPDP instance $(H,P)$, for each pair of nodes $(s,t) \in P$, there exists at least one shortest path in the subgraph $H^*$ connecting the pair of nodes $(s,t)$. Now, for each pair of nodes $(s,t) \in P$ consider one of their corresponding shortest paths in $H^*$ and denote it by $p_{(s,t)}$. As we proved in Lemmas \ref{lemma:shph1} and \ref{lemma:shph2}, for each demand pair $(s_i,t_i) \in A_i$, the shortest path in $H^*$ that connects $s_i$ and $t_i$ is one of the paths $p_{v(i)}: v(i) \in A_i$. Moreover, for each source-sink pair $(o_i,d_i) \in B_i$, the shortest path in $H^*$ that connects $o_i$ and $d_i$ is one of the paths $q_{u(i)}: u(i) \in B_i$. Each of the edges of $H^*$ should be contained in at least one of the shortest paths between demand pairs in $|P|$. As a result, we have $H^*=\bigcup_i p_{v(i)} \bigcup_i q_{v(i)}$. Now, we compute the total savings in $H^*$. First, note that all the paths $p_{v(i)}$ are mutually disjoint and do not have overlap. Also, all the paths $q_{u_{i}}$ are mutually exclusive and do not have overlap. As a result, it is sufficient to compute the savings on the edges that are intersection of a path $p_{v(i)}$ and a path $q_{u(j)}$. Also, note that for each pair of selected paths $p_{v(i)}$ and $q_{u(j)}$ in $H^*$, there exist one edge with weight 1 if and only if $v(i)$ and $u(i)$ are connected in graph $G$ from the MAX REP instance $G(V_1, V_2, E)$. As a result, the total savings can be computed as the number of super-edges covered by vertices $\bigcup_{i=1}^k v(i) \bigcup_{i=1}^k u(i)$, which is a feasible solution to the instance $G(V_1, V_2, E)$.

On the other hand, given a feasible solution $\bigcup_{i=1}^k v(i) \bigcup_{i=1}^k u(i)$ to instance $G(V_1, V_2, E)$. The union of the paths $\bigcup_{i=1}^k p_{v(i)} \bigcup_{i=1}^k q_{u(i)}$ is a feasible solution to the CSPDP instance $(H,P)$. Moreover, the total savings can be computed as the number of super-edges covered by vertices $\bigcup_{i=1}^k v(i) \bigcup_{i=1}^k u(i)$. 

Also, note that number of edges in the constructed graph $H$ is bounded above by $(|V_1|+|V_2|)(2n+2)=O(n^2)$.

Now, we use the result of \cite{arora1997hardness,hochba1997approximation} and the fact that $m$ the number of edges in $H$ is at most $O(n^2)$, to conclude there is no approximation algorithm for CSPDP achieving a ratio $2^{\log^{1-\epsilon}m^{1/2}}$, for any $0< \epsilon< 1$, unless $NP \subseteq DTIME(n^{polylog(n)})$.

This also implies that CSPDP cannot be approximated within $O(m^{1/6-\epsilon})$ factor in polynomial time, unless there is an improvement on the best polynomial time approximation for the $\text{LABEL-COVER}_{\max}$ problem.

\end{proof}

\section*{Acknowledgments}
We are thankful to Saeed Ilchi for the insightful technical discussions and his valuable suggestions on an earlier draft of this paper that has considerably improved its presentation. The work described in this paper was partly supported by research grants from the National Science Foundation (CNS-1837245 and CPS-1837245).

\section*{Appendix I} \label{sec:app1}
\noindent Proof of Lemma \ref{lemma:updc}:
\begin{proof}
Denote by $(x^{TN},y^{TN})$ the optimal solution to the problem \ref{opt:pdpsthin}. Note that $(x^{TN},y^{TN})$ is a feasible solution for problem \ref{opt:pdpsrelax}. We prove $z_2(x^{TN},y^{TN}) \geq \frac{1}{m^{\epsilon}} z_1^{TN}(x^{TN},y^{TN})$. To do so, we can separate the objective function over the edges $e$ that are contained by the local graph of at least one demand pair $(s,t)$ with $b_{(s,t)} \leq \sqrt{m}$ and the edges that does not satisfy this condition.
\begin{equation}\label{eq:division}
    z_1^{TN}(x^{TN},y^{TN})= \sum_{\substack{e:\exists (i,j)| b_{(i,j)}\leq \sqrt{m} \\ e \in E_{(i,j)} \\ e \in E_{TN}} } ( (\sum_{(s,t):e\in E_{(s,t)}} x^{(s,t)}_e) -y_e)c(e) + \sum_{\substack{e:\nexists (i,j)| b_{(i,j)}\leq \sqrt{m} \\ e \in E_{(i,j)} \\ e \in E_{TN}} } (( \sum_{(s,t):e\in E_{(s,t)}} x^{(s,t)}_e) -y_e)c(e)
\end{equation}

The left hand side in \ref{eq:division} equals $z_2(x^{TN},y^{TN})$. In order to prove $z_2(x^{TN},y^{TN}) \geq \frac{1}{m^{\epsilon}} z_1^{TN}(x^{TN},y^{TN})$, it is sufficient to prove the right hand side in \eqref{eq:division} is upper bounded the following expression: 

\begin{equation*}
    (1-\frac{1}{m^{\epsilon}}) z^{TN}_1(x^{TN},y^{TN})
\end{equation*}

Note that $(x^{TN},y^{TN})$ is an optimal solution to the problem \ref{opt:pdpsthin}. As a result, for each edge $e$ satisfying $c(e)> 0$ we have:

\begin{equation*}
    y^{TN}_e = \max_{(s,t):e\in E_{(s,t)}} x^{TN,(s,t)}_e
\end{equation*}

Otherwise We can substitute $y^{TN}_e = \max_{(s,t):e\in E_{(s,t)}} x^{TN,(s,t)}_e$ while keeping the same all other elements of the optimal solution $(\vec{x^{opt}},\vec{y^{opt}})$. Doing so, we obtain a feasible solution to the problem \ref{opt:pdpsthin} that increases the objective value $z_1(\vec{x^{opt}},\vec{y^{opt}})$. This contradicts with our assumption that $(\vec{x^{opt}},\vec{y^{opt}})$ is optimal. As a result, we can rewrite the right hand side in \ref{eq:division} as follows:

\begin{equation*}
    \sum_{\substack{e:\nexists (i,j)| b_{(i,j)}\leq \sqrt{m} \\ e \in E_{(i,j)} \\ e \in E_{TN}} } ( \sum_{(s,t):e\in E_{(s,t)}} x^{TN,(s,t)}_e -y^{TN}_e)c(e)= \sum_{\substack{e:\nexists (i,j)| b_{(i,j)}\leq \sqrt{m} \\ e \in E_{(i,j)} \\ e \in E_{TN}} } ( \sum_{(s,t):e\in E_{(s,t)}} x^{TN,(s,t)}_e -\max_{(s,t):e\in E_{(s,t)}} x^{TN,(s,t)}_e)c(e)
\end{equation*}

The outer summation excludes any edge that is contained in the local graph of a demand pairs $(s,t)$ with $b_{(s,t)} \leq \sqrt{m}$. Hence, we can restrict the inner summation to demand pairs that satisfy $b_{(s,t)}>\sqrt{m}$, while not decreasing the value of the inner summation. Moreover, we can restrict the domain of $\max$ function to decrease its output. As a result:

\begin{equation*}
    \sum_{\substack{e:\nexists (i,j)| b_{(i,j)}\leq \sqrt{m} \\ e \in E_{(i,j)} \\ e \in E_{TN}} } ( \sum_{(s,t):e\in E_{(s,t)}} x^{TN,(s,t)}_e -\max_{(s,t):e\in E_{(s,t)}} x^{TN,(s,t)}_e)c(e) 
\end{equation*}
\begin{equation*}
    \leq \sum_{\substack{e:\nexists (i,j)| b_{(i,j)}\leq \sqrt{m} \\ e \in E_{(i,j)} \\ e \in E_{TN}} } ( (\sum_{\substack{(s,t):e\in E_{(s,t)}\\ b_{(s,t)} > \sqrt{m}}} x^{TN,(s,t)}_e) - \max_{\substack{(s,t):e\in E_{(s,t)}\\ b_{(s,t)} > \sqrt{m}}} x^{TN,(s,t)}_e)c(e)
\end{equation*}
Now, we can expand the domain of the outer summation to obtain:

\begin{equation*}
    \sum_{\substack{e:\nexists (i,j)| b_{(i,j)}\leq \sqrt{m} \\ e \in E_{(i,j)} \\ e \in E_{TN}} } ( (\sum_{\substack{(s,t):e\in E_{(s,t)}\\ b_{(s,t)} > \sqrt{m}}} x^{TN,(s,t)}_e) - \max_{\substack{(s,t):e\in E_{(s,t)}\\ b_{(s,t)} > \sqrt{m}}} x^{TN,(s,t)}_e)c(e) \   
\end{equation*}
\begin{equation*}
    \leq \sum_{e\in E_{TN}} ( (\sum_{\substack{(s,t):e\in E_{(s,t)}\\ b_{(s,t)} > \sqrt{m}}} x^{TN,(s,t)}_e) - \max_{\substack{(s,t):e\in E_{(s,t)}\\ b_{(s,t)} > \sqrt{m}}} x^{TN,(s,t)}_e)c(e)
\end{equation*}

Now, we can use definition \ref{def:light} to conclude:
\begin{equation*}
    \sum_{e\in E_{TN}} ( (\sum_{\substack{(s,t):e\in E_{(s,t)}\\ b_{(s,t)} > \sqrt{m}}} x^{TN,(s,t)}_e) - \max_{\substack{(s,t):e\in E_{(s,t)}\\ b_{(s,t)} > \sqrt{m}}} x^{TN,(s,t)}_e)c(e) \leq  (1-\frac{1}{m^{\epsilon}}) z^{TN}_1(x^{TN},y^{TN})
\end{equation*}
This concludes the proof.
\end{proof}

\section*{Appendix II} \label{sec:app4}
\noindent In this Section we provide an Algorithm to transfer a link-based solution to path-based solution in a reasonable amount of time.

\vspace{0.5cm}
\begin{algorithm}[H]
    \caption{transforming link-based solution to path-based solution}
    \label{alg:transfer}
\vspace{0.5cm}
\begin{algorithmic}

 \STATE 1. For each pair $(s,t)$, consider the edge $e$ with minimum $x^{2,(s,t)}_e$, and find a shortest path $p^e_{(s,t)}$ in $G_{(s,t)}$ from $s$ to $t$ (similar to the procedure in Lemma \ref{thm:main}).
 \STATE 2. Substitute $x^{2,(s,t)}_g$ by $x^{2,(s,t)}_g-x^{2,(s,t)}_e$ for all edges $g$ contained in $p^e_{(s,t)}$
 \STATE 3. Add path $p^e_{(s,t)}$ with weight $f_{p^e_{(s,t)}}=x^{2,(s,t)}_e$ to the set of paths F.
 \STATE 4. If $\vec{x} = \vec{0}$ go to step 5 otherwise go to step 1.
 \STATE 5. Return F.
\end{algorithmic}
\end{algorithm}
\vspace{0.5cm}

\section*{Appendix III} \label{sec:app2}
\noindent Proof of Lemma \ref{lemma:updc}:
\begin{proof}

If $c(e) = 0$, then both sides of \ref{dynmatch} equal zero. It remains to consider the case $c(e) \geq 0$. The variables $x^{opt,(s,t)}_e$ and $y^{opt}_l$ are binary variables. As such, the summation $\sum_{(s,t):e\in E_{(s,t)}} x^{opt,(s,t)}_e$ can be either $0$ or greater than $0$. Let us assume there exists an optimal solution $(\vec{x^{opt}},\vec{y^{opt}})$ with $\sum_{(s,t):e\in E_{(s,t)}} x^{opt,(s,t)}_e=0, y^{opt}_l=1$. We can substitute $y^{opt}_e=0$ while keeping the same all other elements of the optimal solution $(\vec{x^{opt}},\vec{y^{opt}})$. Doing so, we obtain a feasible solution to the problem \ref{opt:pdps} that increases the objective value $z_1(\vec{x^{opt}},\vec{y^{opt}})$ by at least $c(e)$. This contradicts with our assumption that $(\vec{x^{opt}},\vec{y^{opt}})$ is optimal.

As a result it remains to consider the two following cases:
\begin{itemize}
    \item $\sum_{(s,t):e\in E_{(s,t)}} x^{opt,(s,t)}_e=0$ and $y^{opt}_e=0$: In this case both sides of expression \ref{dynmatch} equal 0.
    \item $\sum_{(s,t):e\in E_{(s,t)}} x^{opt,(s,t)}_e\geq 1$: In this case constraint \ref{cons1:vector} leads $y^{opt}_l=1$. As such, 
    \begin{equation*}
        ((\sum_{(s,t):e\in E_{(s,t)}} x^{opt,(s,t)}_e) -y^{opt}_e)v_e c(e) = ((\sum_{(s,t):e\in E_{(s,t)}} x^{opt,(s,t)}_e) -1)v_e c(e)
    \end{equation*}
    On the other hand, in any feasible solution $(\vec{x^{opt}},\vec{y^{opt}})$, the maximum number of shortest paths containing edge $l$, $x^{opt,(s,t)}_e=1$, is less than or equal to the number of local subgraphs containing edge $e$, i.e. $(\sum_{(s,t):e\in E_{(s,t)}} x^{opt,(s,t)}_e) \leq v_e$. Therefore:
    \begin{equation*}
         ((\sum_{(s,t):e\in E_{(s,t)}} x^{*(s,t)}_e) -1)v_e c(e)\leq ( \sum_{(s,t):e\in E_{(s,t)}} x^{(s,t)}_e)(v_e-1)c(e)
    \end{equation*}
    This concludes the proof.
\end{itemize}
\end{proof}

\section*{Appendix IV} \label{sec:app3}
\noindent Proof of Lemma \ref{eq:partleft}:

\begin{proof}
We prove this using induction on the size of set $S$. For the base case $|S|=1$. This yields:
\begin{equation*}
    x_e -1+ (1-x_e)=0
\end{equation*}
Assume for any set $|S|=n$ we have the following:

\begin{equation}
    \sum_{e\in S}x_e-1+\prod_{e\in S} (1-x_e)\geq 0
\end{equation}

Consider a new set $S^*: |S^*|=n+1$ and an arbitrary element $e^* \in S^*$. As a result:
\begin{equation}\label{eq:apundist}
    \sum_{e\in S^*}x_e-1+\prod_{e\in S^*} (1-x_e)= (\sum_{e\in S^* \setminus e^*}x_e-1)+x_{e^*}+(\prod_{e\in S^*\setminus e^*} (1-x_e))(1-x_{e^*})
\end{equation}

We can distribute the multiplication in the right hand side of \ref{eq:apundist} to rewrite \ref{eq:apundist} as follows:
\begin{equation*}
    (\sum_{e\in S^* \setminus e^*}x_e-1)+(\prod_{e\in S^*\setminus e^*} (1-x_e))+x_{e^*}-x_{e^*}(\prod_{e\in S^*\setminus e^*} (1-x_e))
\end{equation*}

On the other hand, we have $\prod_{e\in S^*\setminus e^*} (1-x_e) \leq 1$. Hence:
\begin{equation*}
    \sum_{e\in S^*}x_e-1+\prod_{e\in S^*} (1-x_e)\geq (\sum_{e\in S^* \setminus e^*}x_e-1)+(\prod_{e\in S^*\setminus e^*} (1-x_e)) \geq 0
\end{equation*}

Where the last inequality comes from the induction assumption.

\end{proof}

\section*{Appendix V} \label{sec:app5}
\noindent Proof of Lemma \ref{lemma:shph1}:

If a shortest path $p_{i}$ from $s_i$ to $t_i$ in graph $H$ contains any edges of type $(x_{v_{l_u}}^{2r},d_i)$ or edges of type $(o_i, y_{u_r})$. Then, it contains at least two of them, since the path $p_{i}$ should come back to the graph $H_i^1$. This leads to the shortest path having length at least $(4Nk+1) u_{max}$. However, the paths $p_v: v \in A_i$ have length $(Nk+1) u_{max}+ Nk$ This contradicts with the path $p$ being a shortest path. Moreover, any of the edges $(x_{v_{i}}^{2r},x_{v_{i+1}}^{2r-1})$ increases the distance from $t_i$. As a result, the only shortest paths from $s_i$ to $t_i$ in graph $H$ are the paths $p_v: v \in A_i$.

\section*{Appendix VI} \label{sec:app6}
\noindent Proof of Lemma \ref{lemma:shph2}:

If a shortest path $p_{i}$ from $o_i$ to $d_i$ in graph $H$ contains any edges from ... of weight $u_{max}$ then the length of the path $p_{i}$ is at least $(4Nk+1) u_{max}$. However, the paths $q_u: v \in B_i$ have length $(Nk+2) u_{max}- u_{b_i}>(4Nk+1) u_{max}$ This contradicts with the path $p$ being a shortest path. As a result, The only shortest paths from $o_i$ to $d_i$ in graph $H$ are the paths 
$q_u$.

\bibliographystyle{plainnat}
\bibliography{refs}

\end{document}